\documentclass[11pt,a4paper]{article}
\usepackage{filecontents}

\usepackage{tikz}
\usepackage[english]{babel}
\usepackage{natbib,hyperref}
\usepackage{epsfig}
\usepackage{graphicx}
\usepackage{amsfonts,amsmath,amsthm,amssymb}
\newtheorem{theorem}{Theorem}%
\newtheorem{proposition}[]{Proposition}
\newtheorem{definition}{Definition}
\newtheorem*{remark}{Remark}
\usepackage{tikz}
\usepackage{color,hyperref}
\usepackage[margin=1.2in]{geometry}
\numberwithin{equation}{section}
\usetikzlibrary{fit,
  matrix,
  positioning,
  decorations,
  decorations.pathreplacing
}

\usepackage{latexsym}

\usepackage{xparse}
\usepackage{float,soul,ulem,cancel}
\providecommand*\email[1]{\href{mailto:#1}{#1}}
\DeclareMathOperator{\Tr}{Tr}
\newcommand*{\doi}[1]{\href{http://dx.doi.org/#1}{doi: #1}}

\def\&{\vspace{-5pt}&}

\def\blue#1{\textcolor[rgb]{0,0,1}{#1}}

\usepackage[]{youngtab}

\usepackage{xcolor}
\definecolor{MyBlue}{rgb}{0.25,0.5,0.75}
\colorlet{NextBlue}{MyBlue!20}
\colorlet{SecondBlue}{MyBlue!40}

\NewDocumentCommand{\tens}{t_}
 {%
  \IfBooleanTF{#1}
   {\tensop}
   {\otimes}%
 }
\NewDocumentCommand{\tensop}{m}
 {%
  \mathbin{\mathop{\otimes}\displaylimits_{#1}}%
 }

\begin{document}


\title{The $\tau$-function of the Ablowitz-Segur family of solutions to Painlev\'e  II as a Widom constant}

\author{ H. Desiraju\footnote{{\email{harini.desiraju@sissa.it}}}
  \\
  \normalsize\it Scuola Internazionale Superiore di Studi Avanzati,   \\
  \normalsize\it Via Bonomea, 265, 34136 Trieste, Italy.} 

\date{\vspace{-5ex}}
\maketitle

\begin{abstract}
 $\tau$-functions of certain Painlev\'e equations (PVI,PV,PIII) can be expressed as a Fredholm determinant. Further, the minor expansion of these determinants provide an interesting connection to Random partitions. 
 This paper is a step towards understanding whether the $\tau$-function of Painlev\'e II has a Fredholm determinant representation. The Ablowitz-Segur family of solutions are special one parameter solutions of Painlev\'e II and the corresponding $\tau$-function is known to be the Fredholm determinant of the Airy Kernel. We develop a formalism for open contour in parallel to the one formulated in \cite{CGL} in terms of the Widom constant and verify that the Widom constant for Ablowitz-Segur family of solutions is indeed the determinant of the Airy Kernel. Finally, we construct a suitable basis and obtain the minor expansion of the Ablowitz-Segur $\tau$-function.
 \end{abstract}

 \tableofcontents

\section{Introduction}
Painlev\'e equations are nonlinear second order ODEs  whose solutions are widely recognized  as important special functions with a broad range of applications. The integrability property of these equations was obtained by representing them as an isomonodromic system  of ordinary differential equations. The   Riemann-Hilbert  (RH)  method then proved to be a powerful technique to study  solutions and their properties. An important object related to the solutions is the so called $\tau$-function.
 
In the theory of isomonodromic deformations, the $\tau$-function ($\tau_{JMU}$) was introduced by the Kyoto school  and it is  constructed starting from a certain $1$-form $\omega_{JMU}$ on the space of the deformation parameters \cite{JMU}. If  the parameters are of isomonodromic type, then the form $\omega_{JMU}$ is closed  with respect to differentiation with respect to the parameters. The corresponding $\tau_{JMU}$ function is defined  locally as 
 \begin{equation}
 \label{1.1}
 d\log\tau_{JMU}=\omega_{JMU}
 \end{equation}
 where $d$ denotes total differentiation with respect to the parameters.
 A notable example  of  $\tau$-function of the Painlev\'e II equation is the Tracy-Widom distribution  \cite{TW}. 
  Such $\tau$- function has the property of being expressed as a Fredholm determinant of the Airy kernel. Many relevant solutions of the Painlev\'e equations that appear in various branches of mathematics  turn out to be expressed as a Fredholm determinant of some integrable operator, as defined by Its, Izergin, Korepin, and Slavnov \cite{IIKS}.
  For example the gap probability distribution in random matrices is  Fredholm determinant with the   sine kernel (Painlev\'e V)  \cite{TW}, the  correlation function   of  stochastic point processes on a one-dimensional lattice originated from  representations of the infinite symmetric group is a Fredholm determinant with hypergeometric kernel (Painlev\'e VI) \cite{BD},\cite{BO}. 
  
  It is natural to inquire whether  all solutions of Painlev\'e equations can be expressed as a Fredholm determinant of some integrable operator.
In a series of recent papers Cafasso, Gavrylenko, Lisovyy \cite{GL}, \cite{CGL}  showed  that the  generic  $\tau$-function  of the  PVI, PV, PIII equation can be expressed as a Fredholm determinant.
The key feature of this construction is to reduce the Riemann-Hilbert problem  (RHP) associated   to the isomonodromic system to a RHP on the circle  for a jump  matrix $G$.
Then one can define a  Toeplitz operator $T_G=\Pi_+G$ where $G$ is the jump (called \textit{symbol} in the literature of Toeplitz determinants) of the RHP and $\Pi_+$ the  projection operator to analytic functions in the interior of the circle.
It has been shown in \cite{GL},\cite{CGL} that the  Fredholm determinant 
\begin{equation}
\label{tau_intro}
\tau[G]= \det\left( T_{G^{-1}} \circ T_G \right),
\end{equation}   coincides  up to a factor  with the isomonodromic $\tau$ function \eqref{1.1}.
The above $\tau$ function is also called Widom constant  since such quantity was obtained by Widom in the description of the asymptotic behaviour of Toeplitz determinants \cite{Widom1}, \cite{Widom2}, when the size of the matrix tends to infinity, as a refinement of the strong Szeg\"o theorem.


This approach  cannot be directly implemented to  the cases where the RHP is formulated on a contour that is not a circle as is the case for the  Painlev\'e equations PI, PII and  PIV. 
There are several examples of  $\tau$-functions expressed   as a  Fredholm determinant  like the  Tracy-Widom distribution related to  Painlev\'e II \cite{TW},  or like the example obtained in \cite{BGR} related to 
Painlev\'e IV. 
However the generic $\tau$-function  of the Painlev\'e I, II and IV equations does not seem to have a Fredholm determinant representation.
 The main obstacle to develop the procedure implemented in \cite{GL} and \cite{CGL} is   the impossibility to reduce the RHP problem of the Painlev\'e I, II and IV equations to a RHP on the circle.
However it is expected that the generic RHP for these equations could be reduced to a RHP on the line for a jump matrix $G$.
Then one considers the  projection operator $\Pi_+$  to holomorphic functions on the semi-plane  and define the operator  $T_G=\Pi_+G$.
For the case considered  in this manuscript,  this operator    is trace class and  therefore the  Fredholm determinant \eqref{tau_intro}
is well defined.
In this manuscript  we develop this machinery for the Painlev\'e II equation by considering  the Ablowitz-Segur family of solutions for PII as a toy model.
This example  serves as a starting point to obtain the generic Painlev\'e II $\tau$-function.
 \vspace*{0.2cm}

 This paper is structured as follows. We will first setup the machinery to extend the formalism in \cite{GL},\cite{CGL}  to a line contour and show in \ref{theorem:theorem1} that the $\tau$-function can be written as a Widom constant.  Next we show in \ref{proposition:proposition3} that the Widom constantcoincides with the isomonodromic $\tau$-function \eqref{1.1}.
 Finally in \ref{proposition:proposition4} we construct an appropriate basis and study the minor expansion of the Widom constant.

\section{Setup}
Let $J(z,t): i\mathbb{R} \rightarrow SL(2,\mathbb{C}) $ be a smooth matrix function of $z$ depending analytically on the complex parameter $t$ in some domain. We assume that $||J(z)- 1|| = \mathcal{O}(|z|^{-1})$ as $z \rightarrow \pm  i\infty$. We shall refer to $J$ as the \textit{jump matrix}. In association with the data of the contour $(i\mathbb{R})$ and jump matrix $J$ one can introduce two \textit{Riemann Hilbert problems}, also known as \textit{factorization problems}. They consist of two $2\times 2$ matrices $\Theta(z,t)$ and $\Psi(z,t)$ such that:

\begin{itemize}
  \item $\Theta(z,t)$, $\Psi(z,t)$ are analytic in $z \in \mathbb{C}/i\mathbb{R}$ and admit continuous boundary values from the left $(+)$ and right $(-)$ sides of $i\mathbb{R}$.
  \item The boundary values $\Theta_{\pm}$, $\Psi_{\pm}$ satisfy the \textit{jump conditions}
  \begin{equation}
  J(z,t)= \Theta_{-}^{-1}(z,t)\Theta_{+}(z,t) = \Psi_{+}^{-1}(z,t)\Psi_{-}(z,t)  \,\, ;\,\,\, z\in i\mathbb{R} \label{rhp}
\end{equation}
  \item The functions $\Theta(z,t)$ and $\Psi(z,t)$ are normalized at infinity.
  \begin{equation}
  \lim_{z \rightarrow \infty} \Theta(z,t) = I \,;\, \lim_{z \rightarrow \infty} \Psi(z,t) = I
  \end{equation}
  where the limit is intended as limit in any proper subsector of the left/right half-planes.
\end{itemize}  

The two solutions (if they exist) are the two \textit{Birkhoff factorizations}; we stipulate to call $\Theta(z,t)$ the \textit{direct} Riemann-Hilbert problem and $\Psi(z,t)$ the \textit{dual}. The contour $i\mathbb{R}$ divides the complex plane  into the right half (negative side) and the left half (positive side). The space $L^2\left( i\mathbb{R},|dz| \right)\otimes \mathbb{C}^2$ can be split as the direct sum of two closed subspaces ({\it Hardy spaces}):
 $$\mathcal{H} = L^2(i\mathbb{R}, \mathbb{C}^2) = \mathcal{H}_{+} \oplus \mathcal{H}_{-}$$ the functions on $\mathcal{H}$ are all column vectors. The two subspaces consist of (vector valued) functions in $L^2(i\mathbb{R})$ that are boundary values from the left($+$)/right($-$) of analytic functions that tend to 0 at infinity. Notice that this splitting is orthogonal. On these spaces, one can define projection operators $\Pi_{\pm}$ such that 

$$\Pi_{+} : \mathcal{H} \rightarrow \mathcal{H}_{+} \quad ; \quad \Pi_{-}: \mathcal{H} \rightarrow \mathcal{H}_{-}$$
explicitly, $\Pi_{\pm}$ are just the Cauchy transforms
\begin{align}
\Pi_{+} f(z) = \int_{i\mathbb{R}} \frac{dw}{2\pi i}  \frac{f(w)}{w-z}  \qquad \Re z<0 \\ \label{2.3}
\Pi_{-} f(z) = -\int_{i\mathbb{R}}\frac{dw}{2\pi i} \frac{f(w)}{w-z}   \qquad \Re z>0
\end{align}
 with the equality $\Pi_{+} + \Pi_{-} \equiv \mathbb{I}$. To define the $\tau$-function, we  first  define the operator $T_{J^{-1}}:\mathcal{H} \rightarrow \mathcal{H}_{+}$ for the symbol $J^{-1}$ by multiplication followed by projection
\begin{equation}
T_{J^{-1}} (f) = \Pi_{+}( J^{-1} f) \label{2.4}
\end{equation}
$T_{J}$ is similarly defined 
\begin{equation}
T_{J} (f) = \Pi_{+} \left( J f \right).
\end{equation}

With the operators $T_J$ and $T_{J^{-1}}$, we define the $\tau$-function along the same lines as the Widom constant.  
\begin{definition}
We define the Widom constant with respect to the operators $T_{J}$, $T_{J^{-1}}$
\begin{equation}
\tau [J] = \det \left( T_{J^{-1}} \circ T_{J} \right) \label{2.5}
\end{equation}
\end{definition}

It is a Fredholm determinant for the case of Ablowitz-Segur solutions as will be shown later. 

\begin{proposition}\label{prop1}
 $\tau[J]$ as defined in \eqref{2.5} admits an equivalent representation as the determinant 
\begin{equation}
\tau[J] = \det\textsubscript{$\mathcal{H}$}[1 + U] \label{FD}
\end{equation}
where ${\bf 1}$ denotes the identity operator on $\mathcal{H}$, $U: \mathcal{H} \rightarrow \mathcal{H}$ is an operator represented in the splitting $\mathcal{H}_{\pm}$ as
$U = \left(  \begin{array}{cc} 
0 & a \\
b & 0 
\end{array} \right) 
$ and $a:\mathcal{H}_{-} \rightarrow \mathcal{H}_{+}$; $ b: \mathcal{H}_{+} \rightarrow \mathcal{H}_{-}$ are given by,
\begin{equation*}
a = \Theta_{+}  \Pi_{+} \Theta_{+}^{-1} -\Pi_{+} \quad; \quad b = \Pi_{-}- \Theta_{-}  \Pi_{-} \Theta_{-}^{-1}. 
\end{equation*} 

\end{proposition}

\begin{proof}

Substituting \eqref{2.4} in \eqref{2.5} and manipulating the terms gives the familiar form of the determinant representation of $\tau$-function in \cite{CGL}. 
\begin{eqnarray}
\tau [J] = \det\textsubscript{$\mathcal{H}_{+}$}[ T_{J^{-1}} \circ T_{J} ] = \det\textsubscript{$\mathcal{H}_{+}$} [\Pi_{+} J^{-1} \Pi_{+} J] \nonumber \\
=   \det\textsubscript{$\mathcal{
 H}_{+}$}[\Pi_{+} \Theta_{+}^{-1} \Theta_{-} \Pi_{+} \Theta_{-}^{-1} \Theta_{+} ] \nonumber \\
= \det\textsubscript{$\mathcal{
 H}_{+}$}[\Theta_{+}\Pi_{+} \Theta_{+}^{-1} \Theta_{-} \Pi_{+} \Theta_{-}^{-1}  ] \nonumber \\
 =\det\textsubscript{$\mathcal{
 H}_{+}$}[\Theta_{+}\Pi_{+} \Theta_{+}^{-1} \Theta_{-} (1- \Pi_{-}) \Theta_{-}^{-1}  ] \nonumber \\
 = \det\textsubscript{$\mathcal{
 H}_{+}$}[ 1- (\Theta_{+}\Pi_{+} \Theta_{+}^{-1}) ( \Theta_{-} \Pi_{-} \Theta_{-}^{-1})  ] \nonumber \\
 = \det\textsubscript{$\mathcal{
 H}_{+}$}[ 1- (\Theta_{+}\Pi_{+} \Theta_{+}^{-1} - \Pi_{+}) (\Pi_{-}- \Theta_{-} \Pi_{-} \Theta_{-}^{-1})  ] \nonumber \\
= \det\textsubscript{$\mathcal{H}$}[1 + U]  
\end{eqnarray}
where 
\begin{equation}
U = \left(  \begin{array}{cc} 
0 & a \\
b & 0 
\end{array} \right) \label{U}
\end{equation}
and 
\begin{equation}
a = \Theta_{+}  \Pi_{+} \Theta_{+}^{-1} -\Pi_{+} \quad; \quad b = \Pi_{-}-\Theta_{-}  \Pi_{-} \Theta_{-}^{-1}. \label{a_b}
\end{equation} 
Notice that $$a: \mathcal{H}_{-} \rightarrow \mathcal{H}_{+} \quad ; \quad  b: \mathcal{H}_{+} \rightarrow \mathcal{H}_{-}.$$ 
\end{proof}
Now, one can repeat the same computation as above in terms of the dual RHP $\Psi_{\pm}$ and get the following. 
 \begin{equation}
a = \Psi_{+}  \Pi_{+} \Psi_{+}^{-1} -\Pi_{+} \quad b = \Pi_{-}- \Psi_{-}  \Pi_{-} \Psi_{-}^{-1}. \label{2.10}
\end{equation}
\section{Toy model: Ablowitz-Segur solution} 

The Ablowitz-Segur family \cite{AS}
 of solutions  of  the  Painlev\'e II   equation \begin{equation}
\label{PII}
\partial_{s}^2 u = su + 2u^{3} 
\end{equation}
are specified uniquely by the boundary condition 
\begin{equation}
 \label{3.1}
u(s) \simeq \kappa Ai(s); \quad s\rightarrow +\infty, \quad \kappa\in \mathbb{C}\,.
\end{equation}
It is well known that the solution of a Painlev\'e equation can be obtained by solving a Riemann-Hilbert problem for a matrix valued  function  \cite{BIK} 
$\Gamma(z,s)$ analytic in the complex $z$ domain $\mathbb{C}$ minus some contours.  \begin{figure}[H]
\centering
\includegraphics{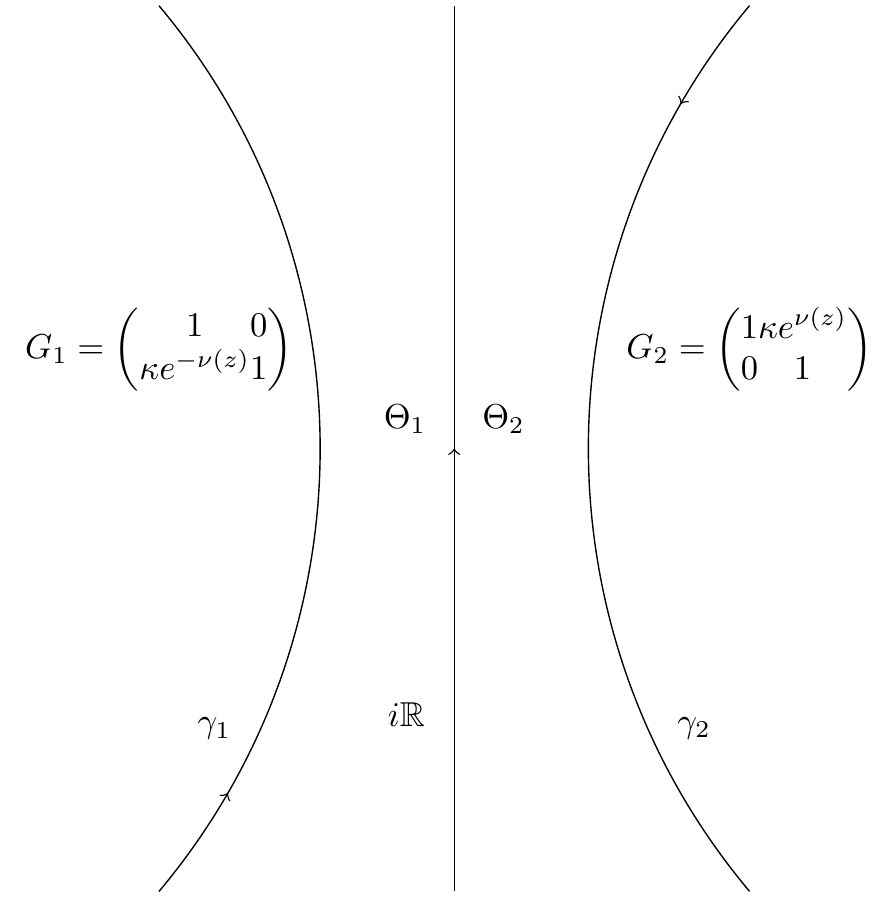}
\caption{Contour \label{Contour 1}}
\end{figure}
For the Ablowitz-Segur family of solutions,  the contours  are shown in figure~\ref{Contour 1},
 and the  Riemann-Hilbert problem satisfied by $\Gamma(z,s)$ is as follows
 \begin{itemize}
 \item  $\Gamma(z,s)$ is  is analytic in $\mathbb{C} \backslash \Sigma$ with  $\Sigma \equiv \gamma_{1}\cup \gamma_{2}$;
 \item the  boundary values $\Gamma_{\pm}(z,s)$ on the oriented contours $\gamma_1$ and $\gamma_2$  satisfying the following jump conditions
\begin{eqnarray}
 \quad \Gamma_{+}(z,s) = G_{1}(z,s)\Gamma_{-}(z,s)  \quad z\in \gamma_{1} \label{rhp_gamma1}\\
\quad \Gamma_{+}(z,s) = G_{2}(z,s)\Gamma_{-}(z,s) \quad z \in \gamma_{2} \label{rhp_gamma2} 
\end{eqnarray}
\item  the asymptotic behaviour at infinity is specified by 
\begin{eqnarray}
\Gamma(z,s) =  {\bf 1} + \dfrac{\Gamma^{(1)}(s)}{z}+O(z^{-2}),\quad \mbox{as} \;\;|z|\to\infty .
\end{eqnarray}
\end{itemize}
The solution $\Gamma(z,s)$ of the  above  Riemann Hilbert problem  (when it exists)   determines the Painlev\'e transcendent  via the relation
\begin{equation}
u(s) = 2 \Gamma^{(1)}_{12}(s). 
\end{equation}
From the above data, the isomonodromic $\tau$-function is defined by \cite{JMU}
\begin{equation}
\partial_{s} \log \tau_{JMU}(s) = - Res_{z=\infty } \Tr \left[ \Gamma(z,s)' \Gamma^{-1}(z,s) \left(z\sigma_{3} \right) \right] \label{JMU_PII}
\end{equation}
and we have  the relation
\[
u^2(s)=-\dfrac{d^2}{ds^2} \log \tau_{JMU}(s).
\]
Instead of solving the Riemann Hilbert problem for $\Gamma$ we factorize it into two separate Riemann-Hilbert problems, one for the function $\Theta_1(z,s)$  analytic in  $z\in \mathbb{C} \backslash \gamma_1$ and 
one for the function $\Theta_2(z,s)$  analytic in  $z\in \mathbb{C} \backslash \gamma_2$  with boundary values
\begin{eqnarray}
\Theta_{1+} (z,s) = G_{1}(z,s)\Theta_{1-}(z,s) \quad z\in \gamma_{1} \label{rhp_theta1} \\
\quad \Theta_{2+} (z,s) = G_{2}(z,s) \Theta_{2-} (z,s) \quad z\in \gamma_{2} \label{rhp_theta2} 
\end{eqnarray}
and 
\begin{eqnarray}
\Theta_{1}(z,s) =  {\bf 1} + \mathcal{O}(z^{-1}) \,,\quad \Theta_{2}(z,s) =  {\bf 1} + \mathcal{O}(z^{-1}) \quad  \mbox{as $z\rightarrow \infty$}.
\end{eqnarray}

It is straightforward to solve the Riemann Hilbert problems \eqref{rhp_theta1}, \eqref{rhp_theta2}. The solutions are given by the Cauchy transforms of the respective jumps $G_{1}$, $G_{2}$
\begin{eqnarray}
\Theta_{1} (z, s)= \left[ \begin{array}{ccc} 
1 & & 0 \\
\kappa \int_{\gamma_1} \frac{e^{-\nu(\lambda,s)}}{\lambda - z} \frac{d\lambda}{2\pi i} &  & 1   \end{array} \right] \label{3.3} \\
\Theta_{2}(z,s) = \left[  \begin{array}{ccc} 
1 & & \kappa \int_{\gamma_2} \frac{e^{\nu(\lambda,s)}}{\lambda - z} \frac{d\lambda}{2\pi i} \\
0 & & 1 \end{array} \right]. \label{3.4}
\end{eqnarray}
Next we define the matrix valued function $\Theta(z,s)$ such that
\begin{equation}
\Theta(z,s) = \begin{cases}
\Theta_{2}(z,s) \quad \Re z<0 \\
\Theta_{1}(z,s) \quad \Re z\geq 0 
\end{cases} \label{def_theta}
\end{equation} 
Clearly the matrix function  $\Theta(z,s)$ is analytic in $\mathbb{C}\backslash  i\mathbb{R} $  and the boundary values $\Theta_{\pm}$ on $i\mathbb{R} $ satisfy   the  jump condition
{\small \begin{equation}
\Theta_{-}(z,s)^{-1} \Theta_{+}(z,s)= J(z,s)= \left[\begin{array}{ccc}
1 & & \kappa \int_{\gamma_2} \frac{e^{\nu(\lambda,s)}}{\lambda - z} \frac{d\lambda}{2\pi i} \\
-\kappa \int_{\gamma_1} \frac{e^{-\nu(\lambda,s)}}{\lambda - z} \frac{d\lambda}{2\pi i} & & 1 
\end{array} \right] \, . \label{rhp_theta}
\end{equation}}
and 
\[
 \lim_{z\rightarrow \infty} \Theta(z,s) = {\bf 1}. 
 \]

\subsection{Computing the $\tau$-function}
In this section we want to make sense of the quantity $\det\left(T_{J^{-1}} \circ T_{J} \right)$   introduced in \eqref{2.5}   when the matrix $J$ is as in  \eqref{rhp_theta}.
In Proposition~\ref{prop1} we show that 
\begin{equation}
\tau[J] =\det\left(T_{J^{-1}} \circ T_{J} \right)= \det\textsubscript{$\mathcal{H}$} \left[ {\bf 1} + U \right] \label{FD_AS}
\end{equation}
where ${\bf 1}$ denotes the identity operator on $\mathcal{H}$,  and 
$U = \left(  \begin{array}{cc} 
0 & a \\
b & 0 
\end{array} \right) $  with  $a:\mathcal{H}_{-} \rightarrow \mathcal{H}_{+}$; $ b: \mathcal{H}_{+} \rightarrow \mathcal{H}_{-}$  given by,
\begin{equation*}
a = \Theta_{+}  \Pi_{+} \Theta_{+}^{-1} -\Pi_{+} \quad; \quad b = \Pi_{-}- \Theta_{-}  \Pi_{-} \Theta_{-}^{-1} 
\end{equation*}
with $ \Theta_{\pm}$ the boundary values of the matrix $\Theta$  defined in \eqref{def_theta}.

 We want to show that the  quantity  \eqref{FD_AS}  is a Fredholm determinant and coincides with the $\tau$-function defined in  \cite{B1}.


We remind the reader of the result in \cite{B1} where  the $\tau$-function of  the Ablowitz-Segur family of solutions for Painlev\'e II is given by the following  Fredholm determinant\footnote[1]{$\gamma_{\pm}$ are $\gamma_{1,2}$ rotated by $\pi/2$ and this is also the source of the factor of $i$ in the exponential in \cite{B1}.}

\begin{equation}
\tau(s) = \det \left[ Id_{L^2(\gamma_{+} \cup \gamma_{-})} - \kappa \left[ \begin{array}{cc}  0 & \mathcal{F} \\ \mathcal{G} & 0  \end{array}   \right] \right] = \det \left[ Id_{L^2(\gamma_{+})} - \kappa^2 \mathcal{F} \circ \mathcal{G}   \right] \label{tau_cite3}
\end{equation}
with 
\begin{equation}
\mathcal{F} : L^2(\gamma_{-}) \rightarrow L^2(\gamma_{+}) \quad \mathcal{G} : L^2(\gamma_{+}) \rightarrow L^2(\gamma_{-}) 
\end{equation}

and 
\begin{eqnarray} 
(\mathcal{F}g)(z) = e^{-\frac{i}{2} \nu(z,s)} \int_{\mathbb{R}-ic} \frac{dw}{2\pi i} \frac{e^{\frac{i}{2} \nu(w,s)} g(w)}{w-z}  \\
(\mathcal{G}g)(z) = e^{\frac{i}{2} \nu(z,s)} \int_{\mathbb{R}+ic} \frac{dw}{2\pi i} \frac{e^{-\frac{i}{2} \nu(w,s)} g(w)}{w-z} 
\end{eqnarray}

 \begin{theorem}\label{theorem:theorem1}
 
 The $\tau$-function  \eqref{tau_cite3} of the  Ablowitz-Segur family of solutions of the Painlev\'e II equation is the Widom constant defined in  \eqref{FD_AS}.
\end{theorem}

\begin{proof}

The  Widom constant can be obtained from \eqref{FD_AS} by computing   the operators $a$, $b$ explicity. Let  $f(z)\in \mathcal{H}_{-}$, $h(z)\in \mathcal{H}_{+}$  be  vector valued
functions
$$f \equiv \left( \begin{array}{c}
f_1 \\ f_2
\end{array} \right); \,\, h \equiv \left( \begin{array}{c}
h_1 \\ h_2
\end{array} \right), $$
then 
\begin{eqnarray}
af(z) =  \int_{i\mathbb{R}}\frac{dw}{2\pi i} \frac{ \Theta_{2}(z) \Theta_{2}^{-1} (w)-1 }{w-z} f(w)  \label{3.6} \\
bh(z) =  \int_{i\mathbb{R}} \frac{dw}{2\pi i} \frac{1 -\Theta_{1}(z) \Theta_{1}^{-1} (w)  }{w-z} h(w),  \label{3.5}
\end{eqnarray}
with $\Theta_1$ and $\Theta_2$ are as in \eqref{3.4}.
We begin by computing \eqref{3.6}
\begin{eqnarray}
 \Theta_{2}(z) \Theta_{2}^{-1}(w) -1 = \left[ \begin{array}{cccc} 
0 & & & \kappa \int_{\gamma_2}\left(  \frac{e^{\nu(\lambda,s)}}{\lambda - z} - \frac{e^{\nu(\lambda,s)}}{\lambda - w} \right) \frac{d\lambda}{2\pi i} \\
0 & & & 0   \end{array} \right]\label{3.10}
\end{eqnarray}
substituting \eqref{3.10} in \eqref{3.5} and focusing on the only non-zero entry $a_{12}$,

\begin{equation}
a_{12} f_{2}(z) = -\kappa \int_{i \mathbb{R}}\frac{dw}{2\pi i} \int_{\gamma_{2}} \frac{d\lambda}{2\pi i}\frac{e^{\nu(\lambda,s)}}{(\lambda - z)(\lambda - w)} f_{2}(w)   \label{3.11}
\end{equation}
integrating over $\lambda$,
\begin{equation}
a_{12} f_{2}(z) = -\kappa \int_{i \mathbb{R} - \epsilon}\frac{dw}{2\pi i}  \frac{e^{\nu(w,s)}}{w- z} f_{2}(w)  \label{3.12}
\end{equation}
A similar computation for $b$ gives that the only non-zero entry is $b_{21}$ that reads
\begin{equation}b_{21} h_{1}(z) =  \kappa \int_{i \mathbb{R}} \int_{\gamma_{1}} \frac{e^{-\nu(\lambda,s)}}{(\lambda - z)(\lambda - w)} h_{1}(w) \frac{d\lambda}{2\pi i} \frac{dw}{2\pi i} \label{3.13}
\end{equation}
integrating over $\lambda$
\begin{equation}
b_{21} h_{1}(z) = \kappa \int_{i \mathbb{R}+\epsilon}  \frac{e^{-\nu(w,s)}}{(w - z)} h_{1}(w)  \frac{dw}{2\pi i}  \label{3.14}
\end{equation}
Substituing $a$ and $b$ back in \eqref{FD_AS}, we get the following

\begin{eqnarray}
\tau(s) = \det \left[ Id_{L^2(i\mathbb{R})\otimes \mathbb C^2} -  \left[ \begin{array}{cc}  0 & a \\ b & 0  \end{array}   \right] \right] 
= \det \left[ Id_{L^2(i\mathbb{R})\otimes \mathbb C^2}  -  \left[
 \begin{array}{cccc}  
 0 & 0  & 0 & a_{12} \\ 
 0 & 0  & 0&  0 \\ 
 0 & 0   &0&  0 \\ 
 b_{21} & 0  &0&   0  \end{array}   \right] 
 \right] \\
 =\det \left[ Id_{L^2(i\mathbb{R})}  -  \left[
 \begin{array}{cc}  
 0  & a_{12} \\ 
 b_{21}  &  0 \end{array}   \right] 
 \right]. \label{AS_tau}
\end{eqnarray}

 Further, it is straightforward to see that the operator $\left[ \begin{array}{cc}  0 & a \\ b & 0  \end{array}   \right] $ is trace class. In other words, $a_{12}$ and $b_{21}$ are Hilbert-Schmidt
{\footnotesize \begin{eqnarray}
|a_{12}|^2 = -\kappa^2 \int_{i\mathbb{R} +\epsilon} |dz| \int_{i\mathbb{R}-\epsilon} |dw| \frac{e^{\nu(w,s)+\nu(\bar{w})}}{|w-z|^2} = -\kappa^2 \int_{i\mathbb{R} +\epsilon} |dz| \int_{i\mathbb{R}-\epsilon} |dw| \frac{e^{2Re\nu(w,s)}}{|w-z|^2} < +\infty \label{4.101} \\
|b_{21}|^2 = -\kappa^2 \int_{i\mathbb{R} -\epsilon} |dz| \int_{i\mathbb{R}+\epsilon} |dw| \frac{e^{-\nu(w,s)-\nu(\bar{w})}}{|w-z|^2} = -\kappa^2 \int_{i\mathbb{R} -\epsilon} |dz| \int_{i\mathbb{R}+\epsilon} |dw| \frac{e^{-2Re\nu(w,s)}}{|w-z|^2} < +\infty \label{4.100} 
\end{eqnarray}} 
\eqref{4.100} and \eqref{4.101} are clearly convergent, implying that $a_{21}$ and $b_{12}$ are Hilbert-Schmidt operators. Therefore, the determinant $\det_{\mathcal{H}} [1+U]$ is Fredholm and coincides with the $\tau$-function in \eqref{tau_cite3}.
\end{proof}
 \begin{remark}Further, it is shown in \cite{BC} that 
\begin{equation}
\det\left[ Id_{L^2\left(\gamma_{L}\cup \gamma_{R}\right)} + U \right] = \det\left[ Id_{L^2\left([s,\infty) \right)} - \kappa^2 K_{Ai} \vert_{[s,\infty)}   \right] \label{4.102}
\end{equation}
where $K_{Ai}$ is the Airy kernel, which implies the $\tau$-function \eqref{AS_tau} is the determinant of the Airy Kernel. It is a well known result \cite{TW} that the solution of \eqref{3.1} is related to the Airy Kernel as 
\begin{equation}
u(s)^2 = -\frac{d^2}{ds^2} \log\det \left[ {\bf 1} - \kappa^2 K_{Ai} \vert_{[s,\infty)} \right]
\end{equation}
\end{remark}

\subsection{Relation to the JMU tau-function}
The logarithmic derivative of the Widom constant  \eqref{FD_AS} can be shown to coincide with the logarithmic derivative of the isomonodromic $\tau$-function \eqref{JMU_PII}. To this  end, we begin by defining a matrix valued function $Y(z,s)$ as a ratio of $\Theta(z,s)$ and $\Gamma(z,s)$.
\begin{eqnarray}
Y(z,s) =\begin{cases} \Theta_2^{-1}(z,s) \Gamma(z,s) \quad \Re z<0 \\
\Theta_1^{-1}(z,s)\Gamma(z,s) \quad \Re z \geq 0  \label{def_Y}
\end{cases}
\end{eqnarray}

$Y(z,s)$ has a jump only on $i\mathbb{R}$. Its boundary values satisify the following relation
\begin{equation}
Y_{+}(z,s) = J^{-1}(z,s)Y_{-}(z,s) \quad z\in i\mathbb{R} \label{rhp_Y}
\end{equation}
and $Y(z,s) \rightarrow {\bf 1}$ as $z\rightarrow \infty$.

\textbf{Notation:} $' \equiv \frac{\partial}{\partial z}$ and $ \dot{} \equiv \frac{\partial}{\partial s}$. All functions depend on $z$,$s$ unless stated otherwise.

\begin{proposition}\label{proposition:proposition2}
The logarithmic derivative of Widom constant  in \eqref{FD_AS} is
\begin{equation}
\partial_{s} \log \tau[J] = \int_{i\mathbb{R}} \frac{dz}{2\pi i} \Tr \left\lbrace J^{-1} \dot{J} \left[ -Y_{+}' Y_{+}^{-1} + \Theta_{2}^{-1} \Theta_{2}'  \right]  \right\rbrace\,. \label{int.rep}
\end{equation}
\end{proposition}
\begin{proof}
Begin with the Fredholm determinant \footnote[2]{This computation follows from Theorem 2.3 in \cite{CGL}. The difference being the choice of factorisation.}
\begin{equation}
\tau[J] = \det \left[  T_{J^{-1}} \circ T_{J} \right]= \det\left[PQ\right]
\end{equation}
where $P=T_{J^{-1}} = \Pi_{+} J^{-1}$ and $Q= T_{J} = \Pi_{+}J$. Inverses are $P^{-1} = Y_{-} \Pi_{+} Y_{+}^{-1}$ and $Q^{-1} = \Theta_{+}^{-1} \Pi_{+} \Theta_{-}$. Computing the logarithmic derivative
\begin{align}
\partial_{t} \log \det \left[ PQ \right] = \Tr \left[ \partial_{t} P P^{-1} + Q^{-1} \partial_{t} Q \right] \nonumber\\
= \Tr \left[ -\Pi_{+} J^{-1} \partial_{t} J J^{-1} Y_{-} \Pi_{+} Y_{+}^{-1} + \Theta_{+}^{-1} \Pi_{+} \Theta_{-} \Pi_{+ } \partial_{t} J \right] \nonumber\\
=\Tr \left[ -\Pi_{+} J^{-1} \partial_{t} J \,\, Y_{+} Y_{-}^{-1} Y_{-} \Pi_{+} Y_{+}^{-1} + \Theta_{+}^{-1} \Pi_{+} \Theta_{-} \left(1-\Pi_{- } \right)\partial_{t} J \right] \nonumber\\
=\Tr \left[ -\Pi_{+} J^{-1} \partial_{t} J \,\, Y_{+} \Pi_{+} Y_{+}^{-1} + \Theta_{+}^{-1} \Pi_{+} \Theta_{-} \partial_{t} J \right] \nonumber\\
=\Tr \left[ -\Pi_{+} J^{-1} \partial_{t} J \,\, Y_{+} \Pi_{+} Y_{+}^{-1} + \Theta_{+}^{-1} \Pi_{+} \Theta_{+} J^{-1} \partial_{t} J \right] \nonumber \\
=\int_{i\mathbb{R}} \frac{dz}{2\pi i} \Tr \left\lbrace J^{-1} \partial_{t} J \left[ - \partial_{z} Y_{+} Y_{+}^{-1} + \Theta_{+}^{-1} \partial_{z} \Theta_{+} \right] \right\rbrace
\end{align}
to obtain the last expression  we use the fact that $J^{-1}\partial_{t} J$ is a multiplication operator and only the diagonal parts of $Y_{+} \Pi_{+} Y_{+}^{-1}$ and $\Theta_{+}^{-1} \Pi_{+} \Theta_{+}$ contribute to the expression. 
\end{proof}

Using the above proposition we can identify the Widom constant with the isomonodromic  $\tau$ function.

\begin{proposition} \label{proposition:proposition3}
The logarithmic derivative of the  Widom  constant \eqref{int.rep} coincides exactly with the logarithmic derivative of the (isomonodromic) JMU $\tau$- function \eqref{JMU_PII} for Ablowitz-Segur family of solutions namely :
\begin{equation}
\partial_{s} \log \tau[J] = \partial_{s} \log \tau_{JMU} = - Res_{z=\infty} \Tr \left[ \Gamma' \Gamma^{-1} \left( z \sigma_{3} \right) \right]\,.
\end{equation}
\end{proposition}

\begin{proof}  We prove the statement by simplifying the expression of $\partial_{s} \log \tau[J]$ in  \eqref{int.rep}:
\begin{eqnarray}
\Tr \left\lbrace J^{-1} \dot{J} \left[  -Y_{+}' Y_{+}^{-1} + \Theta_{2}^{-1} \Theta_{2}' \right]  \right\rbrace \label{3.26}.
\end{eqnarray}
We first perform algebraic manipulation on   $ Y_{+}' Y_{+}^{-1}$ using  \eqref{def_Y}  
\begin{eqnarray}
Y_{+}' Y_{+}^{-1} = \left(\Theta_{2}^{-1} \Gamma \right)' \left(\Gamma^{-1} \Theta_{2} \right) \nonumber \\
= \left(-\Theta_{2}^{-1} \Theta_{2}' \Theta_{2}^{-1} \Gamma + \Theta_{2}^{-1} \Gamma' \right) \left(\Gamma^{-1} \Theta_{2} \right) \nonumber \\
=  \Theta_{2}^{-1} \left(- \Theta_{2}' \Theta_{2}^{-1} +  \Gamma' \Gamma^{-1} \right)\Theta_{2}  \label{3.28}
\end{eqnarray}
and expressing $J$ in terms of $\Theta_{1}$ and $\Theta_{2}$ we obtain 
\begin{eqnarray}
J^{-1}\dot{J}  = (\Theta_{2}^{-1} \Theta_{1}) \dot{(\Theta_{1}^{-1} \Theta_{2})} \nonumber \\
= (\Theta_{2}^{-1} \Theta_{1}) (-\Theta_{1}^{-1}\dot{\Theta}_{1} \Theta_{1}^{-1} \Theta_{2} + \Theta_{1}^{-1} \dot{\Theta}_{2}) \nonumber \\
= \Theta_{2}^{-1}  (-\dot{\Theta}_{1} \Theta_{1}^{-1}  +  \dot{\Theta}_{2} \Theta_{2}^{-1} ) \Theta_{2} \nonumber \\
= \Theta_{2}^{-1}  \Delta(\dot{\Theta} \Theta^{-1}   ) \Theta_{2} \label{3.29}
\end{eqnarray}
where 
\begin{equation}
\Delta (\dot{\Theta} \Theta^{-1}) = \dot{\Theta}_{2} \Theta_{2}^{-1} -\dot{\Theta}_{1} \Theta_{1}^{-1}. \nonumber
\end{equation}
Substituting \eqref{3.28} and \eqref{3.29} in \eqref{3.26} and using cyclicity of trace,
\begin{eqnarray}
\Tr \left\lbrace J^{-1} \dot{J} \,\, Y_{+}' Y_{+}^{-1}  \right\rbrace \ = \Tr \left\lbrace \Delta \left( \dot{\Theta} \Theta^{-1}  \right)  \left( \Gamma' \Gamma^{-1}  + \Theta_{2}' \Theta_{2}^{-1}  \right)  \right\rbrace \label{3.30} 
\end{eqnarray}
The term $  -\Delta \left( \dot{\Theta} \Theta^{-1} \right)\Theta_{2}' \Theta_{2}^{-1} $ is explicit and cancels the term $J^{-1} \dot{J} \Theta_{2}' \Theta_{2}^{-1}  $. After the simplification, \eqref{int.rep} is 
\begin{equation}
-\int_{i\mathbb{R}} \frac{d z}{2\pi i} \Tr \left\lbrace \Delta \left( \dot{\Theta} \Theta^{-1} \right)  \Gamma' \Gamma^{-1}     \right\rbrace \label{3.31}
\end{equation}
since $\Gamma$ has no jump on $i\mathbb{R}$, \eqref{3.31} can be further simplified
\begin{eqnarray}
-\int_{i\mathbb{R}} \frac{d z}{2\pi i} \Tr  \left\lbrace \Delta \left( \dot{\Theta} \Theta^{-1} \right)  \Gamma' \Gamma^{-1}    \right\rbrace = -\int_{i\mathbb{R}} \frac{d z}{2\pi i} \Tr \Delta  \left\lbrace  \dot{\Theta} \Theta^{-1}  \Gamma' \Gamma^{-1}  \right\rbrace \nonumber \\
 = \int_{\Sigma} \frac{dz}{2\pi i} \Tr \Delta  \left\lbrace  \dot{\Theta} \Theta^{-1}  \Gamma' \Gamma^{-1}  \right\rbrace \nonumber \\
 = \int_{\gamma_1} \frac{dz}{2\pi i} \Tr \Delta  \left\lbrace  \dot{\Theta} \Theta^{-1}  \Gamma' \Gamma^{-1}  \right\rbrace + \int_{\gamma_{2}} \frac{dz}{2\pi i}  \Tr \Delta  \left\lbrace \dot{\Theta} \Theta^{-1}  \Gamma' \Gamma^{-1}  \right\rbrace \label{3.32} 
\end{eqnarray}
Let us begin by computing the integral on $\gamma_1$ in \eqref{3.32}
\begin{eqnarray}
 \Tr \Delta\left\lbrace  \dot{\Theta} \Theta^{-1}  \Gamma' \Gamma^{-1} \right\rbrace = \Tr \left\lbrace  \dot{\Theta}_{1+} \Theta^{-1}_{1+}  \Gamma_{+}' \Gamma_{+}^{-1}  - \dot{\Theta}_{1-} \Theta^{-1}_{1-}  \Gamma_{-}' \Gamma_{-}^{-1} \right\rbrace  \label{3.33}
 \end{eqnarray}
computing \eqref{3.33} term by term by substituting \eqref{rhp_gamma1} for $\Gamma_{+}$
\begin{eqnarray}
 \Gamma_{+}'\Gamma_{+}^{-1}  = (G_{1}\Gamma_{-} )' (\Gamma_{-}^{-1} G_{1}^{-1}) \nonumber \\
 = G_{1} \left[ G_{1}^{-1} G_{1}'  +  \Gamma_{-}' \Gamma_{-}^{-1} \right]G_{1}^{-1} \label{3.34}
\end{eqnarray}
and \eqref{rhp_theta1} for $\Theta_{1+}$
\begin{eqnarray}
\dot{\Theta}_{1+} \Theta_{1+}^{-1} = \dot{\left( G_{1} \Theta_{1-} \right)} \Theta_{1-}^{-1} G_{1}^{-1}  \nonumber \\
= G_{1} \left[ G_{1}^{-1} \dot{G}_{1} + \dot{\Theta}_{1-} \Theta_{1-}^{-1} \right]G_{1}^{-1} \label{3.35}.
\end{eqnarray}
Substituting \eqref{3.34}, \eqref{3.35} in \eqref{3.32} and using cyclicity
\begin{eqnarray}
\Tr \left\lbrace \left( \dot{\Theta}_{1+} \Theta_{1+}^{-1}  \right) \Gamma_{+}' \Gamma_{+}^{-1}  \right\rbrace = \Tr \left[ \left( G_{1}^{-1} \dot{G}_{1} + \dot{\Theta}_{1-} \Theta_{1-}^{-1} \right) \left( G_{1}^{-1} G_{1}'  +  \Gamma_{-}' \Gamma_{-}^{-1}  \right)  \right] \label{3.36}
\end{eqnarray}
In \eqref{3.36}, notice that the term 
 $\left(G_{1}^{-1}  \dot{G}_{1}  +  \dot{\Theta}_{1-}  \Theta_{1-}^{-1}  \right) \left( G_{1}^{-1} G_{1}' \right)$ is traceless. Furthermore, we have the following identity 
 $ 2 G_1^{-1} \dot{G}_{1} = - z G_{1}^{-1} \sigma_3 G_{1} +  z \sigma_3 $. The terms $ \dot{\Theta}_{1-}  \Theta_{1-}^{-1} \Gamma_{-}^{-1} \Gamma_{-}'$ in \eqref{3.36} and \eqref{3.33}  cancel each other out.
So, all that is left to compute on the contour $\gamma_{1}$ is the following
\begin{equation}
\int_{\gamma_1}  \frac{dz}{2\pi i} \Tr \left[ G_{1}^{-1} \dot{G}_{1}\Gamma_{-}' \Gamma_{-}^{-1}  \right] = \frac{1}{2} \int_{\gamma_1}  \frac{d z}{2\pi i} \Tr \left[ \left(  z G_{1}^{-1} \sigma_3 G_{1} + z \sigma_3 \right) \Gamma_{-}' \Gamma_{-}^{-1}   \right] \label{3.38}
\end{equation} 
We begin by computing the following term
\begin{eqnarray}
\Tr \left( G_{1}^{-1} \sigma_{3} G_{1} \Gamma_{-}' \Gamma_{-}^{-1}  \right) \label{3.39}
\end{eqnarray}
the term $\Gamma_{-}'  \Gamma_{-}^{-1}$ can be simplified by substituting \eqref{rhp_gamma1}
\begin{eqnarray}
\Gamma_{-}'  \Gamma_{-}^{-1} = \left( G_{1}^{-1} \Gamma_{+} \right)' \left(  \Gamma_{+}^{-1} G_{1} \right) \nonumber \\
=  G_{1}^{-1} \left(  -G_{1}' G_{1}^{-1}  + \Gamma_{+}' \Gamma_{+}^{-1} \right) G_{1} \label{3.40}
\end{eqnarray}
substituting \eqref{3.40} in \eqref{3.39} and using the cyclic property of the trace
\begin{eqnarray}
\Tr \left(  G_{1}^{-1} \sigma_{3} G_{1}  \Gamma_{-}^{-1}  \Gamma_{-}'\right) = \Tr \left[ \sigma_3 \Gamma_{+}' \Gamma_{+}^{-1}  - G_{1}' G_{1}^{-1} \sigma_3 \right] \label{3.41}
\end{eqnarray}
note that $G_{1}' G_{1}^{-1} \sigma_3$ is traceless. Substituting \eqref{3.41} in \eqref{3.38} we have
\begin{equation}
\frac{1}{2} \int_{\gamma_1}  \frac{dz}{2\pi i} \Tr \left[\left(  z G_{1}^{-1} \sigma_3 G_{1} -  z \sigma_3 \right) \Gamma_{-}' \Gamma_{-}^{-1}   \right] = \frac{1}{2} \int_{\gamma_1}  \frac{dz}{2\pi i} \Tr \left[-  z \sigma_3 \left( \Gamma_{+}' \Gamma_{+}^{-1} - \Gamma_{-}' \Gamma_{-}^{-1}  \right) \right] \label{3.42}
\end{equation}
repeating this exercise and computing the integral on $\gamma_{2}$ in \eqref{3.30}, we get exactly the same expression. Putting all together
\begin{eqnarray}
\partial_s \ln \tau[J] = \int_{\Sigma} \frac{d z}{2\pi i} \Tr \left[ -z \sigma_3  \left( \Gamma_{+}' \Gamma_{+}^{-1} - \Gamma_{-}' \Gamma_{-}^{-1} \right)    \right] \nonumber \\
= -Res_{z=\infty}\blue{ \Tr} \left(  z \sigma_3 \Gamma' \Gamma^{-1} \right) \label{3.43}
\end{eqnarray}
\end{proof}

\section{Minor expansion}
The Hilbert space $L^2(S^1)$ admits a natural orthonormal basis of Fourier modes (i.e. the monomials $z^n, \  n \in \mathbb Z$).  The minor expansion of the Fredholm determinant \eqref{tau_intro} in this particular basis gives rise to interesting combinatorics. In the case of Painlev\'e VI, V, III the combinatorics correspond to certain Nekrasov Partition functions of certain Gauge theories \cite{GL}. 

In this spirit, we would like to propose, at least, a reasonable expansion of the Fredholm determinant of our operator in a similar guise. In our case the underlying Hilbert space $L^2(i\mathbb R)$  does not immediately suggest a natural discrete orthonormal basis. Here below we want to propose a very natural such basis: the main guiding principle is that of identifying the Hardy space $\mathcal H_+$ with the Hardy space of the interior of the disk, and pulling back the monomial basis.  

\begin{proposition}\label{proposition:proposition4}
The Fredholm determinant of the $\tau$-function in \eqref{FD_AS} can be expanded, on an appropriate basis, in terms of minors, that can be labelled by Maya diagrams ($m_{X}$)
\begin{equation}
\tau [s] = \sum_{m_X \in \mathbb{M};\, \, \vert p \vert = \vert h \vert } \alpha_{p_X}^{h_X} \beta_{h_X}^{p_X}
\end{equation}
where the coefficients $\alpha_{m}^{n}$, $\beta_{n}^{m}$ are as follows 
\begin{equation}
\beta_{0}^{0} = \alpha_{0}^{0} = \left(4\frac{\partial^2}{\partial s^2} - 1  \right)  \left( 1- \frac{\partial}{\partial s} \right)^{-1} Ai(s)
\end{equation}
where $Ai(s)$ is the Airy function, and
\begin{eqnarray}
\alpha_{m}^{n} = \frac{(-1)^{m+n}}{(m!)^2 n! (m+n+1)!} \left( \widetilde{D} \right)^{m+n} \alpha_{0}^{0} \\
\beta_{n}^{m} = \frac{(-1)^{m+n}}{(n!)^2 m! (m+n+1)!} \left( \widetilde{D} \right)^{m+n} \beta_{0}^{0} 
\end{eqnarray}
with $\widetilde{D}= 2 \left( \frac{\partial}{\partial s} -1 \right)^2  \left(4 \frac{\partial^2}{\partial s^2} -s \right)$
\end{proposition}

\begin{proof}

Recall that, 
\begin{equation}
\tau(s) = \det \left[ Id_{L^2(i\mathbb{R})} -  \left[ \begin{array}{cc}  0 & a \\ b & 0  \end{array}   \right] \right]
=\det \left[ Id_{L^2(i\mathbb{R})\otimes \mathbb C^2}  -  \left[
 \begin{array}{cccc}  
 0 & 0  & 0 & a_{12} \\ 
 0 & 0  & 0&  0 \\ 
 0 & 0   &0&  0 \\ 
 b_{21} & 0  &0&   0  \end{array}   \right] 
 \right] 
 \label{5.1}
\end{equation}
of which $a_{12}$ and $b_{21}$ are the only non zero entries. Therefore, the  determinant of the $4 \times 4$ block operator can be reduced to a determinant of a $2\times 2$ block operator. Let us denote $a_{12}\equiv \alpha$, $b_{21}\equiv \beta$
\begin{equation}
\tau[s] = \det\left[ Id_{L^2(i\mathbb{R})}- \left( \begin{array}{cc}
0 & a_{12} \\ b_{21} & 0 
\end{array}  \right)  \right] = \det\left[ Id_{L^2(i\mathbb{R})}- \left( \begin{array}{cc}
0 & \alpha \\ \beta & 0 
\end{array}  \right)  \right] \label{5.2}
\end{equation}
We remind the reader here that the block decomposition is due to the splitting $L^2(i\mathbb R) = \mathcal H = \mathcal H_+\oplus \mathcal H_-$.

The first step to obtain the minor expansion is constructing a suitable basis to expand $\alpha(z,w)$ and $\beta(z,w)$,
\begin{equation}
\alpha(z,w): \mathcal{H}_{-} \rightarrow \mathcal{H}_{+} \quad ; \quad \beta(z,w) : \mathcal{H}_{+} \rightarrow \mathcal{H}_{-}. \label{5.3}
\end{equation}

\subsection*{Basis construction} \label{subsec:basis}
The spaces $\mathcal{H}_{\pm}$ are Hardy spaces of functions analytic on the left and right half of the complex planes respectively. To construct the bases of $\mathcal{H}_{\pm}$, we employ the Paley-Weiner theorem which identifies $\mathcal H_+$ as the image under Fourier transform of functions supported on a half--line. Specifically, let  $\mathbb C_{+} = \left\lbrace z: z=x+iy , y>0  \right\rbrace$, 
\begin{equation}
H^2(\mathbb C_{+}) = \left\lbrace f:f\, \textrm{is analytic in} \, \,\mathbb C_{+} \,\, \textrm{and} \sup_{0<y<+\infty}\int_{-\infty}^{+\infty} \vert f(z) \vert^2 dx < \infty    \right\rbrace \label{5.4}.
\end{equation}
By definition, the boundary values of $f\in H^2(\mathbb C_+)$ on $\mathbb R$ define a function in $L^2(\mathbb R)$ and we can think of $H^2(\mathbb C_{+}) $ as a (closed) subspace of $L^2(\mathbb R)$. With this understanding, the Paley--Wiener theorem can be stated as the following identity:
\begin{equation}
\mathcal{F} H^2(\mathbb C_{+}) = L^2[0,\infty) \label{5.5}.
\end{equation}
The space $\mathcal{H}_{+}$  can be isometrically mapped to $H^2(\mathbb C_{+})$ by a variable change $z \rightarrow iz$. 
We have that Laguerre functions  ($L_{n}^{\lambda}(z) z^\lambda {\rm e}^{-z} $) provide a basis of $L^2(\mathbb R_+)$. Using the Paley--Wiener theorem, upon inverse Fourier transform, they yield a basis for $H^2(\mathbb C_{+})$ and an innocent change of variable $ z \rightarrow iz$ gives a basis on $\mathcal{H}_{+}$. We can comfortably restrict ourselves to $\lambda =0$. Following \cite{Shen}, the Fourier transform $(\hat{\ell}_{n}^{\lambda}(t))$ of the Laguerre functions $L_{n}^{\lambda}(x){\rm e}^{-\frac x 2} x^\frac \lambda 2$ for $\lambda=0$, and using the notation and $L_{n}^{0} \equiv L_{n}$, is 
\begin{gather}
L_{n}(x) {\rm e}^{-\frac x 2} = \frac{{\rm e}^{-\frac x 2}}{n!} \sum_{k=0}^n \frac{(-n)_{k} x^k}{k!}   \label{5.6} \\
\hat{\ell}_{n} (t) = \frac{-2}{n!} \left( \frac{1+2it}{2it-1}  \right)^n \frac{1}{2it-1} \label{5.7}
\end{gather}
$\hat{\ell}_{n}$ forms a complete basis on $H^2(\mathbb C_{+}, dt)$. With the change of variable $2it =z$, \eqref{5.3} reads
\begin{equation}
\hat{\ell}_{n}(z) = \frac{-2}{n!} \left( \frac{1+z}{z-1}  \right)^n \frac{1}{z-1} \label{5.8}
\end{equation}
and will form a basis on $H^2(\mathcal{H}_{+}, \frac{-i}{2} dz)$. In conclusion 
\begin{equation}
e_{\mathcal{H}_{+}}^n = \frac{i}{n!} \left( \frac{1+z}{z-1}  \right)^n \frac{1}{z-1} \label{5.9}
\end{equation}
is a basis on $H^2(\mathcal{H}_{+}, dz)$. Similarly,
\begin{equation}
e_{\mathcal{H}_{-}}^n = \frac{i}{n!} \left( \frac{z-1}{z+1}  \right)^n \frac{1}{z+1} \label{5.10}
\end{equation}
is a basis on $H^2(\mathcal{H}_{-}, dz)$

\subsection*{Minor expansion} \label{subsec:minor}
Expanding $\alpha(z,w)$ and $\beta(z,w)$ in the basis $e_{\mathcal{H}_{+}}$ and $e_{\mathcal{H}_{-}}$, the $\tau$-function \eqref{FD_AS} can be expressed as a minor expansion. Starting with $\alpha(z,w)$,
\begin{gather}
\alpha_{m}^{n} = \langle \alpha(z,w) e_{\mathcal{H}_{-}}^{n} , e_{\mathcal{H}_{+}}^{m}(z) \rangle \nonumber \\
= \int_{i\mathbb{R}} \frac{dz}{2\pi i} \bar{e}_{\mathcal{H}_{+}}^m(z) \int_{i\mathbb{R}-\epsilon} \frac{dw}{2\pi i} \alpha(z,w) e_{\mathcal{H}_{-}}^n(w) \nonumber \\
=\frac{-\kappa}{m!n!}\int_{i\mathbb{R}} \frac{dz}{2\pi i} \left(\frac{\bar{z}+1}{\bar{z}-1}\right)^{m} \frac{1}{(\bar{z}-1)} \int_{i\mathbb{R}-\epsilon} \frac{dw}{2\pi i} \frac{e^{\nu(w,s)}}{w-z} \left(\frac{w-1}{w+1}\right)^{n} \frac{1}{(w+1)} \nonumber \\
 =\frac{\kappa}{m!n!}\int_{i\mathbb{R}-\epsilon} \frac{dw}{2\pi i} \frac{e^{\nu(w,s)}(w-1)^{m+n}}{(w+1)^{m+n+2}} \label{5.11}
\end{gather}
Similarly for $\beta(w,z)$
\begin{gather}
\beta_{m}^{n} = \langle \beta(z,w) e_{\mathcal{H}_{+}}^{n} , e_{\mathcal{H}_{-}}^{m}(z) \rangle \nonumber \\
= \int_{i\mathbb{R}} \frac{dz}{2\pi i} \bar{e}_{\mathcal{H}_{-}}^m(z) \int_{i\mathbb{R}+\epsilon} \frac{dw}{2\pi i} \beta(z,w) e_{\mathcal{H}_{+}}^n(w) \nonumber \\
=\frac{\kappa}{m!n!}\int_{i\mathbb{R}} \frac{dz}{2\pi i} \left(\frac{\bar{z}-1}{\bar{z}+1}\right)^{m} \frac{1}{(\bar{z}+1)} \int_{i\mathbb{R}+\epsilon} \frac{dw}{2\pi i} \frac{e^{-\nu(w,s)}}{w-z} \left(\frac{w+1}{w-1}\right)^{n} \frac{1}{(w-1)}  \nonumber \\
 =\frac{-\kappa}{m!n!}\int_{i\mathbb{R}+\epsilon} \frac{dw}{2\pi i} \frac{e^{-\nu(w,s)}(w+1)^{m+n}}{(w-1)^{m+n+2}}  \label{5.12}
\end{gather}

\subsection*{Recurrence relations} \label{subsec:rec.rel}
$\alpha_m^n$, $\beta_{m}^{n}$ can be made explicit by noticing that the functions such as $\int \frac{dw}{2\pi i}(w+1)^m e^{-\theta(w)}$ can be written as some derivatives of the Airy function. Define the function $\chi_{m+n}$ as  
\begin{gather}
\chi_{m+n} = \frac{(w-1)^{m+n}}{(w+1)^{m+n+2}} e^{\nu(w,s)} \label{5.13}.
\end{gather}
Then 
$\alpha_{m}^{n}$ in terms of $\chi_{m+n}$ is simply
\begin{gather}
\alpha_{m}^{n} = \frac{\kappa}{m!n!}\int_{i\mathbb{R}-\epsilon} \frac{dw}{2\pi i} \frac{e^{\nu(w,s)}(w-1)^{m+n}}{(w+1)^{m+n+2}} = \frac{\kappa}{m! n!}\int_{i\mathbb{R}-\epsilon} \frac{dw}{2\pi i} \chi_{m+n} \label{5.14}
\end{gather}
and $\beta_{n}^m$ in terms of $\chi_{m+n}$ is
\begin{gather}
\beta_{n}^{m} = \frac{-\kappa}{m!n!}\int_{i\mathbb{R}+\epsilon} \frac{dw}{2\pi i} \frac{e^{-\nu(w,s)}(w+1)^{m+n}}{(w-1)^{m+n+2}} \nonumber \\
= \frac{\kappa}{m!n!}\int_{i\mathbb{R}-\epsilon} \frac{dw}{2\pi i} \frac{e^{\nu(w,s)}(w-1)^{m+n}}{(w+1)^{m+n+2}} 
 = \frac{\kappa}{m! n!}\int_{i\mathbb{R}-\epsilon} \frac{dw}{2\pi i} \chi_{m+n} \label{5.15}
\end{gather}
The function
$\chi_{m+n}$ obeys a recurrsion relation that can be derived as follows
\begin{gather}
\int \frac{dw}{2\pi i} \chi_{m+n}(w,s) = \int \frac{dw}{2\pi i} \frac{(w-1)^{m+n}}{(w+1)^{m+n+2}} e^{\nu(w,s)} \nonumber \\
 = \frac{2}{(m+n+1)}\int \frac{dw}{2\pi i} e^{\nu(w,s)}\partial_{w} \left[\frac{(w-1)^{m+n+1}}{(w+1)^{m+n+1}} \right] \nonumber \\
= -\frac{2}{(m+n+1)}\int \frac{dw}{2\pi i} \left[\frac{(w-1)^{m+n+1}}{(w+1)^{m+n+1}} \right] \partial_{w} e^{\nu(w,s)} \nonumber \\
=-\frac{2}{(m+n+1)}\int \frac{dw}{2\pi i} \frac{(w-1)^{m+n-1}}{(w+1)^{m+n+1}} (w-1)^{2}  (4w^2 -s) e^{\nu(w,s)} \nonumber \\
 =-\frac{2}{(m+n+1)}\int \frac{dw}{2\pi i} \left( \frac{\partial}{\partial s} -1 \right)^2  \left(4 \frac{\partial^2}{\partial s^2} -s \right) \chi_{m+n-1}(w,s) \label{5.16}
\end{gather}
which gives the following equation 
\begin{gather}
\int_{i\mathbb{R}} \frac{dw}{2\pi i} \left[ \chi_{m+n} + \frac{2}{(m+n+1)} \left( \frac{\partial}{\partial s} -1 \right)^2  \left(4 \frac{\partial^2}{\partial s^2} -s \right) \chi_{m+n-1}\right] =0 \label{5.17}
\end{gather}

We define a function $Ci(t)$  as follows:
\begin{gather}
Ci(t) := \int \frac{dw}{2\pi i} \frac{1}{(w+1)} e^{\theta(w)} = \left( 1- \frac{\partial}{\partial s} \right)^{-1} Ai(t) \label{5.18}
\end{gather}
Then the function
$\chi_{0}$ can be computed in terms of $Ci(t)$
\begin{gather}
\int_{i\mathbb{R}} \frac{dw}{2\pi i} \chi_{0} = \int \frac{dw}{2\pi i} \frac{1}{(w+1)^{2}} e^{\nu(w,s)} 
= -\int \frac{dw}{2\pi i} \partial_{w} \left( \frac{1}{w+1} \right) e^{\nu(w,s)} \nonumber \\
=\int \frac{dw}{2\pi i}  \left( \frac{1}{w+1} \right) \partial_{w} e^{\nu(w,s)} 
= \int \frac{dw}{2\pi i}  \left( \frac{1}{w+1} \right) (4w^2 -s) e^{\nu(w,s)} \nonumber \\
= \int \frac{dw}{2\pi i} \left( 4 \frac{\partial^2}{\partial s^2} - s \right) \left( \frac{1}{w+1} \right) e^{\nu(w,s)} \nonumber \\
= \left( 4 \frac{\partial^2}{\partial s^2} - s \right) Ci(s) \label{5.19}
\end{gather}
Since both  $\alpha_{m}^{n}$ and $\beta_n^m$ are integrals of $\chi_{m+n}$\eqref{5.14}, \eqref{5.15}, the recurrsion relation \eqref{5.17} and the expression for $\chi_{0}$ \eqref{5.19} give the explicit expressions for $\alpha_{m}^{n}$ and $\beta_n^m$. Starting with the recursion relation for $\alpha_{m}^{n}$ and $\beta_n^m$
\begin{eqnarray}
n(m+n+1)\alpha_m^n + 2 \left( \frac{\partial}{\partial s} -1 \right)^2  \left(4 \frac{\partial^2}{\partial s^2} -s \right) \alpha_{m}^{n-1} =0 \label{rec.alpha.n}\\
n(m+n+1)\beta_m^n + 2 \left( \frac{\partial}{\partial s} -1 \right)^2  \left(4 \frac{\partial^2}{\partial s^2} -s \right) \beta_{m}^{n-1} =0 \label{rec.beta.n}
\end{eqnarray}
with 
\begin{equation}
\alpha_{0}^{0} = \left(4\frac{\partial^2}{\partial s^2} - 1  \right)  \left( 1- \frac{\partial}{\partial s} \right)^{-1} Ai(s) \label{zero}
\end{equation}
Further,
\begin{equation}
n \alpha_{m-1}^{n} = m \alpha_{m}^{n-1} \quad ; \quad n\beta_{m-1}^{n} = m\beta_{m}^{n-1} \label{5.22} 
\end{equation}
Notice that \eqref{rec.alpha.n} and \eqref{rec.beta.n} are recurrsive relations in $n$. Using \eqref{5.22} similar recurrsion relations in $m$ can be obtained
\begin{align}
m (m+n+1) \alpha_{m}^{n} + 2 \left( \frac{\partial}{\partial s} -1 \right)^2  \left(4 \frac{\partial^2}{\partial s^2} -s \right) \alpha_{m-1}^{n} =0 \label{rec.alpha.m}\\
m (m+n+1) \beta_{m}^{n} + 2 \left( \frac{\partial}{\partial s} -1 \right)^2  \left(4 \frac{\partial^2}{\partial s^2} -s \right) \beta_{m-1}^{n} =0 \label{rec.beta.m}
\end{align} 
Define $ 2 \left( \frac{\partial}{\partial s} -1 \right)^2  \left(4 \frac{\partial^2}{\partial s^2} -s \right) = \widetilde{D}$. From \eqref{rec.alpha.n} 
\begin{align}
\alpha_{m}^{n} = (-1)^{m+n} \frac{(1+m)}{n! (m+n+1)!} \widetilde{D}^{n} \alpha_{m}^{0}
\end{align}
now using \eqref{rec.alpha.m}
\begin{align}
\alpha_{m}^{0} = (-1)^{m} \frac{1}{m! (m+1)!} \widetilde{D}^m \alpha_{0}^{0}
\end{align}
In terms of $\alpha_{0}^{0}$, $\alpha_{m}^{n}$ is explicit
\begin{align}
\alpha_{m}^{n} = \frac{(-1)^{m+n}}{(m!)^2 n! (m+n+1)!} \widetilde{D}^{m+n} \alpha_{0}^{0} 
\end{align}
Repeating the same computation for $\beta_{n}^{m}$
\begin{align}
\beta_n^{m} = \frac{(-1)^{m+n}}{(n!)^2 m! (m+n+1)!} \widetilde{D}^{m+n} \beta_{0}^{0}
\end{align}

\subsection*{Maya diagrams} \label{subsec:maya}

 The determinant of an operator $A \in \mathbb{C}^{m\times m}$ can be expanded in terms of its principal minors \cite{CGL}. For a finite $m\times m $ matrix $A$, the minor expansion is given by 
\begin{equation}
\det\left( 1+A  \right) = \sum_{n=0}^{\infty} \sum_{i_1< ... <i_n} \det(A_{i_p,i_q})_{p,q=1}^{n} \label{5.23}
\end{equation}
This sequence obviously terminates after $n=m$. Now generalise the minor exansion for infinite dimensional matrix.
\begin{itemize}
\item 
Let $K$ be a infinite dimensional matrix. Instead of labelling by $\left\lbrace 1,...,m \right\rbrace$, $K$ is labelled by an infinite discrete set. Define a half-integer lattice $\mathbb{Z}' = \mathbb{Z}+ \frac{1}{2}$. The set of all finite subsets of $\mathbb{Z}'$ is given by $\left\lbrace 0,1  \right\rbrace^{\mathbb{Z}'}$ and $\chi \subset \left\lbrace 0,1  \right\rbrace^{\mathbb{Z}'}$. For $\mathbb{Z}'_{\mp} = \mathbb{Z}'_{\lessgtr 0}$, define 'particles' ($p_{\chi}$) to be the positions $p_{\chi} = \chi \cap \mathbb{Z}'_{+}$ and 'holes'($h_{\chi}$) to be the positions $h_{\chi}= \chi \cap \mathbb{Z}'_{-}$. $(p_{\chi}, h_{\chi})$ define point configurations on $\mathbb{Z}'$. Furthermore, for a block matrix, we have two indices.
\begin{itemize}
\item Expansion of the block determinant, given by the particles and holes $(p_{\chi},h_{\chi})$

\item The index within the block, which is called the colour index. 
\end{itemize}
\item Maya diagram $m_{\chi}$ is constructed by drawing filled circles at the points $\left( \mathbb{Z}'_{+}\backslash p_{\chi}  \right)\cup h_{\chi}$ and empty circles at $p_{\chi} \cup \left( \mathbb{Z}'_{-}\backslash h_{\chi} \right)$. Set of all Maya diagrams is denoted by $\mathbb{M} = \cup_{\chi} m_{\chi}$.
\item
 For $\det (1+K)$, the minors can be labelled by the half integer lattice $\mathbb{Z}'$. Rows and coloumns will now be labelled by $\chi \subset \left\lbrace 0,1  \right\rbrace^{\mathbb{Z}'}$. The minor expansion is given by 
\begin{equation}
\det\left[1+ K \right] = \sum_{ \chi \subset \left\lbrace 0,1  \right\rbrace^{\mathbb{Z}'}} A_{\chi}
\end{equation}

\item Maya diagrams can also be written as young diagrams by playing the following game. Reading the maya diagram from the left end, draw a horizontal line pointing to the right $(\rightarrow)$ for every filled circle (\tikz\draw[black,fill=black] (0,0) circle (.5ex);) and a vertical line pointing downwards $(\downarrow)$ for every empty circle (\tikz\draw[black] (0,0) circle (.5ex);). 
\end{itemize}

Let us work out a simple example. Let $K$ be a $2\times 2$ block matrix. $A$, $B$ are $2\times 2$ matrices
\begin{equation}
K = \left[ \begin{array}{cc}
0 & A \\ B& 0
\end{array} \right]
\end{equation}
The labelling goes as shown in fig.\ref{minor}. The blue and red label are the colour indices.

\begin{figure}[H]
\centering
\includegraphics{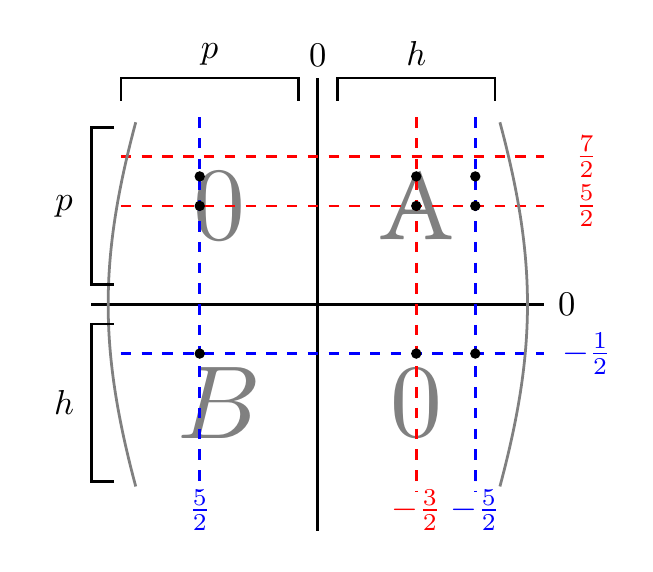}
\caption{Minor expansion \label{minor}}
\end{figure}
 In fig. \ref{minor}, collect the blue and red terms in terms of particles and holes. It is clear from here that minor contribution for the configurations where $|p|\neq |h|$ is $0$. {\color{blue}$p_{\chi} = \left( \frac{5}{2} \right)$}, {\color{blue}$h_{\chi} = \left( -\frac{5}{2},-\frac{1}{2}  \right)$}; {\color{red}$p_{\chi} = \left( \frac{5}{2}, \frac{7}{2} \right)$}, {\color{red}$h_{\chi} = \left( -\frac{3}{2} \right)$}. This gives the blue and red Maya diagrams. {\color{blue} $m_{\chi} = \left( \frac{5}{2},-\frac{5}{2},-\frac{1}{2} \right)$}; {\color{red} $m_{\chi} = \left( \frac{5}{2},\frac{7}{2},-\frac{3}{2} \right)$}

\begin{figure}[H]
\centering
\includegraphics[scale=0.7]{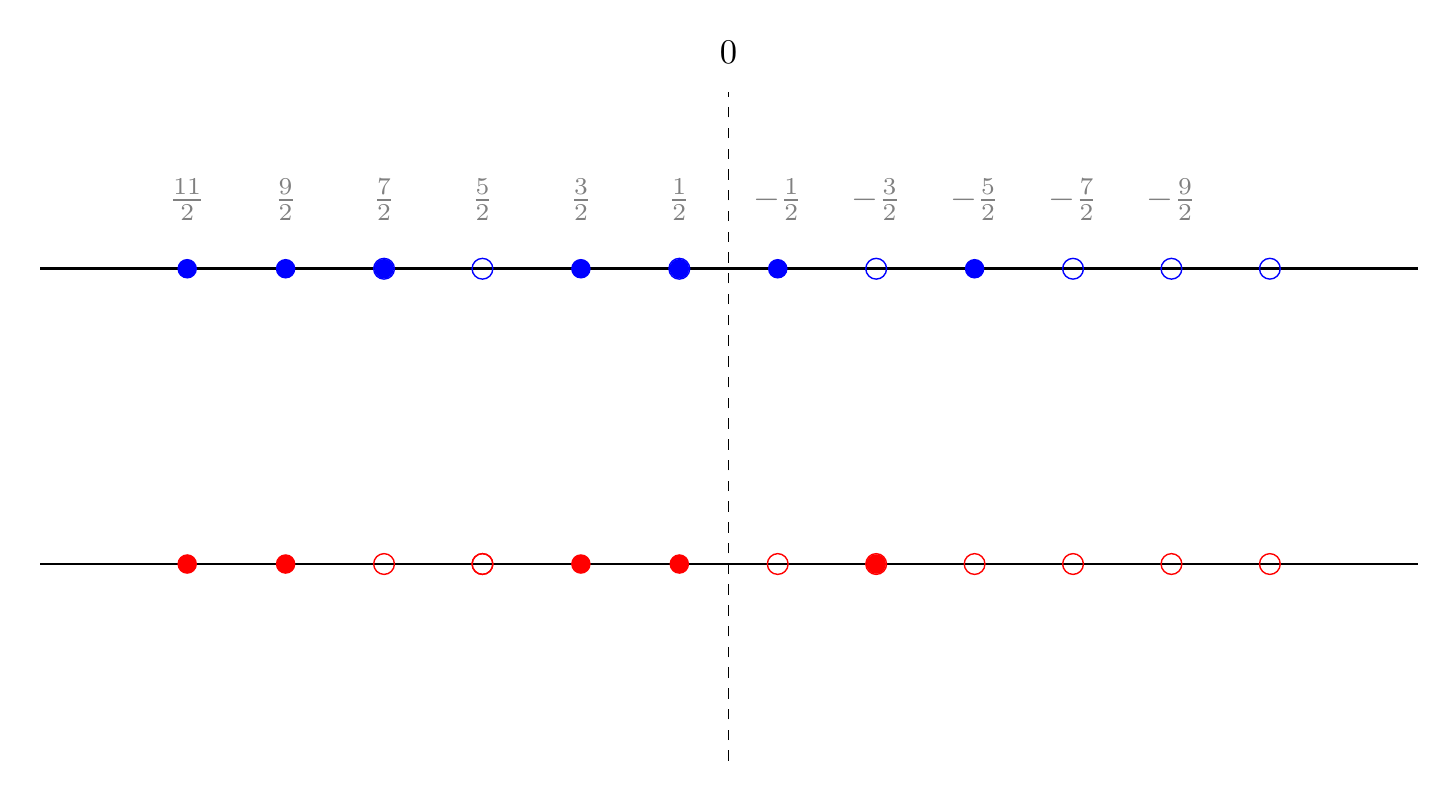}
\caption{Maya diagrams \label{maya}}
\end{figure}
and the corresponding Young tableaux are
{\color{red} \yng(2,2,1)}
{\color{blue} \yng(4,1)}

Now for the Ablowitz-Segur $\tau$-function, 
\begin{equation}
\tau[s] = \det\left[ 1 - \left(\begin{array}{cc}
0 & \alpha \\ \beta & 0
\end{array} \right) \right]
\end{equation}
$\alpha$ and $\beta$ can be expanded on a discrete basis  $e_{\mathcal{H}_{\pm}}(z)$
\begin{gather}
\alpha(z,w) = \sum_{m,n \in \mathbb{Z}_{+}} \alpha_{m}^n e_{\mathcal{H}_{+}}(z)^{m} e_{\mathcal{H}_{-}}(w)^{n} ; \quad \beta(z,w) = \sum_{m,n \in \mathbb{Z}_{+}} \beta_{m}^n e_{\mathcal{H}_{-}}(z)^{m} e_{\mathcal{H}_{+}}(w)^{n} \label{5.25}
\end{gather}

Since $\alpha_m^n$ and $\beta_n^m$ in \eqref{5.25} are not matrices themselves, the corresponding Maya diagrams are "colourless". If $a_m^n$ and $b_n^m$ were $N\times N$ matrices themselves, the corresponding entries in the expansion would be $a_{m;\beta}^{n;\alpha}$ and $b_{n;\alpha}^{m;\beta}$ where $\alpha, \beta = \lbrace 1,...,N \rbrace$ would be the colour indices. Furthermore, given the off-diagonal structure of the matrix $U$, the minors with $\vert p \vert \neq \vert h \vert$ vanish. Therefore, the minor expansion reads, 
\begin{equation}
\tau [s] = \sum_{m_{\chi} \in \mathbb{M};\, \, \vert p \vert = \vert h \vert } \alpha_{p_{\chi}}^{h_{\chi}} \beta_{h_{\chi}}^{p_{\chi}} .\label{5.26}
\end{equation}
The proof is now complete.
\end{proof}

It would be extremely interesting to interpret the terms in this minor expansion in a similar way to the case of Painlev\'e VI, V, III. However, to our knowledge, in the case of the second Painlev\'e transcendent, there is no direct analog  connection with some field theory. Nonetheless the computation proceeds in a  rather natural way and may prove of use in future applications.

\medskip 

\noindent {\bf Acknowledgements.}\\ 
The author thanks Alexander Its, Marco Bertola, Mattia Cafasso and Tamara Grava for illuminating discussions and suggestions. The author acknowledges   the   Centre International de Rencontres Mathematiques, Luminy, for hospitality during the semester "Integrability and Randomness in mathematical Physics" and LYSM grant for financial support. Part of the work was completed during a visit to the Department of Mathematics and Statistics at Concordia University which was supported by H2020-MSCA-RISE-2017 PROJECT No. 778010 IPADEGAN, and the author is greatful for the financial support and hospitality.

\nocite{*}
\bibliographystyle{abbrv}
\bibliography{\jobname}

\end{document}